\documentclass[a4paper,runningheads]{llncs}
\usepackage{lmodern}
\usepackage{amsmath,amssymb,apxproof,balance,enumerate,pgfplots,tikz,tikz-3dplot}

\usepackage[T1]{fontenc}
\usetikzlibrary{calc}
\usetikzlibrary{shapes.geometric}
\usepackage[sort,nocompress]{cite}
\usepackage[ruled,linesnumbered,noend]{algorithm2e}

\newtheorem{rem}{Remark}

\usepackage{subfig}
\theoremstyle{plain}

\def\UIG{\mathrm{UIG}}
\def\APUD{\mathrm{APUD}}
\def\EX{\mathit{ex}}
\newtheorem{defn}{Definition}

\title{On embeddability of unit disk graphs onto straight lines \thanks{This work is supported by the {Czech Science Foundation}, project no.~{20-04567S}.}}

\titlerunning{Axes-parallel unit disk graphs}

\author{Onur \c{C}a\u{g}{\i}r{\i}c{\i}\orcidID{0000-0002-4785-7496}}
\authorrunning{O. \c{C}a\u{g}{\i}r{\i}c{\i}}
\institute{Masaryk University, Brno, Czech Republic\\
\email{onur@mail.muni.cz}}

\begin{document}

\maketitle
\begin{abstract}
	Unit disk graphs are the intersection graphs of unit radius disks in the Euclidean plane.
	Deciding whether there exists an embedding of a given unit disk graph, i.e. unit disk graph recognition, is an important geometric problem, and has many application areas.
	In general, this problem is known to be $\exists\mathbb{R}$-complete.
	In some applications, the objects that correspond to unit disks, have predefined (geometrical) structures to be placed on.
	Hence, many researchers attacked this problem by restricting the domain of the disk centers.
	One example to such applications is wireless sensor networks, where each disk corresponds to a wireless sensor node, and a pair of intersecting disks corresponds to a pair of sensors being able to communicate with one another.
	It is usually assumed that the nodes have identical sensing ranges, and thus a unit disk graph model is used to model problems concerning wireless sensor networks.
	We consider the unit disk graph realization problem on a restricted domain, by assuming a scenario where the wireless sensor nodes are deployed on the corridors of a building.
	Based on this scenario, we impose a geometric constraint such that the unit disks must be centered onto given straight lines.
	In this paper, we first describe a polynomial-time reduction which shows that deciding whether a graph can be realized as unit disks onto given straight lines is NP-hard, when the given lines are parallel to either $x$-axis or $y$-axis.
	Using the reduction we described, we also show that this problem is NP-complete when the given lines are only parallel to $x$-axis (and one another).
	We obtain those results using the idea of the logic engine introduced by Bhatt and Cosmadakis in 1987.
\end{abstract}

\section{Introduction} \label{sec:intro}

An \emph{intersection graph} is a graph that models the intersections among geometric objects.
In an intersection graph, each vertex corresponds to a geometric object, and each edge corresponds to a pair of intersecting geometric objects.
A \emph{unit disk graph} is the intersection graph of a set of unit disks in the Euclidean plane.
Some well-known NP-hard problems, such as chromatic number, independent set, and dominating set, remain hard on unit disk graphs \cite{unitdiskgraphs,udgOptimization,qptas}.
We are particularly interested in the unit disk recognition problem i.e. given a simple graph, deciding whether there exists an embedding of disks onto the plane which corresponds to the given graph.
This problem is known to be NP-hard \cite{udgRecognition}, and even $\exists\mathbb{R}$-complete \cite{sphereAndDotProduct} in general. 

A major application area of unit disk graphs is wireless sensor networks, since it is an accurate model (in an ideal setting) of communicating wireless sensor nodes with identical range \cite{theory,rangebased}.
In a wireless sensor network, the sensor nodes are deployed on bounded areas \cite{cagiricithesis,pathPlanning,connectivityWithBackbone}.
Thus, it becomes more interesting to observe the behavior of the unit disk graph recognition problem when the domain is restricted \cite{udgParameterized,udgIndependentSet,udgConstrained,recognizingDOG}. 

We assume that the sensor nodes are deployed onto the corridors in a building, and the floor plans are available.
We model the corridors on a floor as straight lines, and consider the recognition problem where the unit disks are centered on the given lines.
We show that this problem is NP-hard, even when the given straight lines are either vertical or horizontal, i.e. any pair of lines is either parallel, or perpendicular to each other.
In addition, we show that if there are no pairs of perpendicular lines i.e. all lines are parallel to $x$-axis, then the recognition problem is NP-complete.

Due to space restrictions, the proofs of some statements are omitted, and those statements are marked with (*).

\subsection*{Related work}
Breu and Kirkpatrick showed that the unit disk graph recognition problem is NP-hard in general\cite{udgRecognition}. 
Later on, this result was extended, and it was proved that the problem is also $\exists\mathbb{R}$-complete \cite{sphereAndDotProduct,integerRealization}.
Kuhn et al. showed that finding a ``good'' embedding is not approximable when the problem is parameterized by the maximum distance between any pair of disks' centers \cite{udgApprox}.
In the very same paper, they also give a short reduction that the realization problem and the recognition problem on unit disk graphs are polynomially equivalent \cite{udgApprox}.
	
Intuitively, the most restricted domain for unit disk graphs is when the disks are centered on a single straight line in the Euclidean plane.
In this case, the unit disks become \emph{unit intervals} on the line, and they yield a \emph{unit interval graph} \cite{unitIntervalGraphs}.
To recognize or realize whether a given graph is a unit interval graph is a linear-time task \cite{intervalRecognition}.
Our domain is restricted to not only one straight line, but to a set of straight lines given by their equations.
Given a simple graph, and a set of straight lines, we ask the question ``can this graph be realized as unit disks on the given set of straight lines?''
We show that even though these lines are restricted to be parallel to either $x$-axis or $y$-axis, it is NP-hard to determine whether the given graph can be embedded onto the given lines (Theorem~\ref{thm:main}).
We, however, do not know whether this variant belongs to the class NP, or is possibly $\exists\mathbb{R}$-complete.
If, on the other hand, the lines are restricted to be parallel only to the $x$-axis, then we show that the problem belongs to NP and is still NP-complete (Theorem~\ref{thm:completeness}).

\section{Basic terminology and notations} \label{sec:terminology}
A \emph{unit disk} around a point $p$ is the set of points in the plane whose distance from $p$ is one unit. 
Two unit disks, centered at two points $p$ and $q$, intersect when the Euclidean distance between $p$ and $q$ is less than or equal to two units.
A graph $G = (V,E)$ is called a \emph{unit disk graph} when every vertex $v \in V$ corresponds to a disk $\mathcal{D}_v$ in the Euclidean plane, and an edge $uv \in E$ exists when $\mathcal{D}_u$ and $\mathcal{D}_v$ intersect.

The \emph{unit disk recognition problem} is deciding whether a given graph $G = (V,E)$ is a unit disk graph.
That is, determining whether there exists a mapping $\Sigma: V \to (\mathbb{R} \times \mathbb{R})$, such that each vertex is the center of a unit disk without violating the intersection property. 
The mapping $\Sigma$ is also called the \emph{embedding of $G$ by unit disks}.
We use the domain of \emph{axes-parallel straight lines} which is a set of lines in 2D, where the angle between a pair of lines is either $0$ or $\pi/2$.
This implies that the equation of a straight line is either $y = a$ if it is a horizontal line, or $x = b$ if it is a vertical line, where $a,b \in \mathbb{R}$.
The input for axes-parallel straight lines recognition problem contains two sets, $\mathcal{H},\mathcal{V} \subset \mathbb{R}$, where $\mathcal{H}$ contains the Euclidean distance of each horizontal line from the $x$-axis, and $\mathcal{V}$ contains the Euclidean distance of each vertical line from the $y$-axis.
Thereby in the domain that we use, each vertex is mapped either onto a vertical line, or onto a horizontal line.
We denote the class of axes-parallel unit disk graphs as $\APUD(k,m)$ where $k$ is the number of horizontal lines, and $m$ is the number of vertical lines.
Formally, we define the problem as follows.

\begin{defn}[Axes-parallel unit disk graph recognition on $k$ horizontal and $m$ vertical lines]
The input is a graph $G = (V,E)$, where $V = \{1,2,\dots,n\}$, and two sets $\mathcal{H}, \mathcal{V} \subset \mathbb{Q}$ of rational numbers with where $|\mathcal{H}| = k$ and $|\mathcal{V}| = m$.
The task is to determine whether there exists a mapping $\Sigma: V \to  (\mathbb{R} \times \mathcal{H}) \cup (\mathcal{V} \times \mathbb{R})$ such that there is a unit disk realization of $G$ in which $u 
\in \ell_{\Sigma(u)}$ for each $u \in V$.
\end{defn}

\section{APUD(k,m) recognition is NP-hard} \label{sec:nphard}

We prove that axes-parallel unit disk recognition ($\APUD(k,m)$ recognition with $k$ and $m$ given as input) is NP-hard by giving reduction from the \emph{Monotone not-all-equal 3-satisfiability} (NAE3SAT) problem\footnote{This problem is equivalent to the 2-coloring of 3-uniform hypergraphs. We choose to give the reduction from Monotone NAE3SAT as it is more intuitive to construct for our problem}.
NAE3SAT is a variation of 3SAT where three values in each clause are not all equal to each other, and due to Schaefer's dichotomy theory, The problem remains NP-complete when all clauses are monotone (i.e. none of the literals are negated) \cite{complexitySAT}.
Our main theorem is as follows.

\begin{theorem} \label{thm:main}
	There is a polynomial-time reduction of any instance $\Phi$ of Monotone NAE3SAT to some instance $\Psi$ of $\APUD(k,m)$ such that $\Phi$ is a YES-instance if, and only if $\Psi$ is a YES-instance.
\end{theorem}

We construct our hardness proof using the scheme called a \emph{logic engine}, which is used to prove the hardness of several geometric problems \cite{logicEngine}.
For a given instance $\Phi$ of Monotone NAE3SAT, there are two main components in our reduction.
First, we construct a backbone for our gadget.
The backbone models only the number of clauses and the number of literals.
Next, we model the relationship between the clauses and literals, i.e. which literal appears in which clause.

Let us begin by describing the input graph.
For the sake of simplicity, we assume that the given formula has $3$ clauses, $A, B, C$, and $4$ literals, $q,r,s,t$ for the moment.
In general, we denote the clauses by $C_1, \dots, C_k$, and the literals by $x_1, \dots, x_m$.
Later on, we explain how to generalize the input graph according to any given instance of Monotone NAE3SAT formula.
For the following part, we describe the input graph given in Figure~\ref{fig:skeleton}.
Throughout the manuscript, we index the vertices from left to right, and from bottom to top, in ascending order.

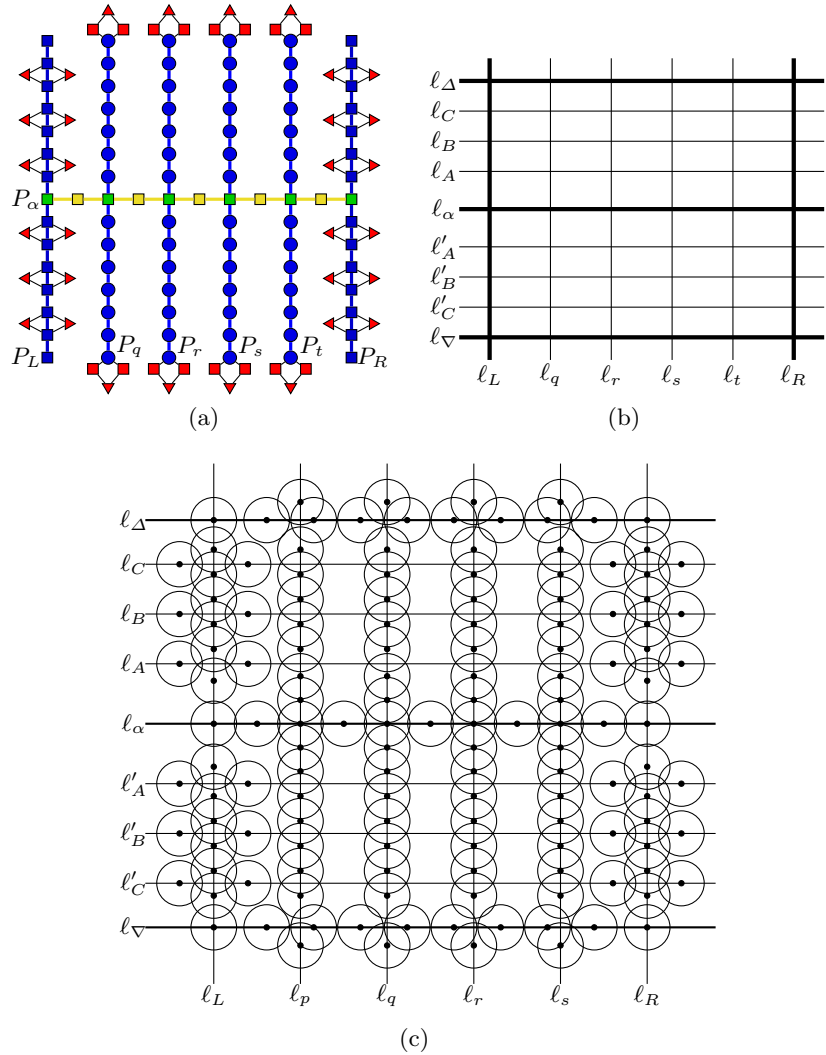
\begin{figure}[htbp]
	\centering
	
	\subfloat[]{
	\centering

	\begin{tikzpicture}[xscale=0.4, yscale=0.3,
	triangle/.style = {regular polygon, regular polygon sides=3},
	ltriangle/.style = {regular polygon, regular polygon sides=3, rotate=90},
	rtriangle/.style = {regular polygon, regular polygon sides=3, rotate=270},
	dtriangle/.style = {regular polygon, regular polygon sides=3, rotate=180}]
	
	\node[draw=none] at (-0.7,-7) {$P_L$};
	\node[draw=none] at (10.7,-7) {$P_R$};
	\node[draw=none] at (-0.7,0) {$P_\alpha$};
	\node[draw=none] at (2.7,-6.5) {$P_q$};
	\node[draw=none] at (4.7,-6.5) {$P_r$};
	\node[draw=none] at (6.7,-6.5) {$P_s$};
	\node[draw=none] at (8.7,-6.5) {$P_t$};

	\tikzstyle{every node}=[draw, shape=circle, minimum size=5pt,inner sep=0pt];

	\foreach \i in {0,2,4,6,8,10} \node[shape=rectangle, fill=green!80!black, minimum size=4pt] (\i0) at (\i,0) {};
	
	\foreach \i in {1,3,5,7,9}	\node[shape=rectangle, fill=yellow!90!black, minimum size=4pt] (\i0) at (\i,0) {};
	
	\foreach \i in {0,...,9}
	{
		\pgfmathtruncatemacro\j{\i+1};
		\draw[very thick, color=yellow!90!black] (\i0)--(\j0);
	}
		
	\foreach \i in {2,4,6,8} 
	{
		\foreach \j in {-7,-6,-5,-4,-3,-2,-1,1,2,3,4,5,6,7}
		{
			\node[fill=blue!90!black] (\i\j) at (\i,\j) {};
			
		}
	}
	
	\foreach \i in {2,4,6,8} 
	{
		\foreach \j in {-7,-6,-5,-4,-3,-2,-1,0,1,2,3,4,5,6}
		{
			\pgfmathtruncatemacro\k{\j+1};
			\draw[very thick, color=blue] (\i\j)--(\i\k);
			
		}
	}

	\foreach \i in {0,10} 
	{
		\foreach \j in {-7,-6,-5,-4,-3,-2,-1,1,2,3,4,5,6,7}
		{
			\node[shape=rectangle,fill=blue!80!black, minimum size=4pt] (\i\j) at (\i,\j) {};
			
		}
	}

	\foreach \i in {0,10} 
	{
		\foreach \j in {-7,-6,-5,-4,-3,-2,-1,0,1,2,3,4,5,6}
		{
			\pgfmathtruncatemacro\k{\j+1};
			\draw[very thick, color=blue!80!black] (\i\j)--(\i\k);
			
		}
	}

	\foreach \i in {-5.5,-3.5,-1.5,1.5,3.5,5.5}
	{
		\pgfmathtruncatemacro\j{\i-0.5}; 
		\pgfmathtruncatemacro\k{\i+0.5};

		\node[ltriangle, fill=red] (-1\j) at (-0.7,\i) {};		
		\node[rtriangle, fill=red] (1\j) at (0.7,\i) {};
		
		\node[ltriangle, fill=red] (9\j) at (9.3,\i) {};
		\node[rtriangle, fill=red] (11\j) at (10.7,\i) {};
		\draw (1\j)--(0\j);
		\draw (1\j)--(0\k);
		\draw (-1\j)--(0\j);
		\draw (-1\j)--(0\k);
		\draw (9\j)--(10\j);
		\draw (9\j)--(10\k);
		\draw (11\j)--(10\j);
		\draw (11\j)--(10\k);
	}

	\foreach \i in {2,4,6,8} 
	{
		\node[triangle, fill=red] (\i9) at (\i,8.3) {};
		\node[dtriangle, fill=red] (\i-9) at (\i,-8.3) {};
	}
	
	\foreach \i in {1.5,2.5,3.5,4.5,5.5,6.5,7.5,8.5} 
	{
		\pgfmathtruncatemacro\j{\i+0.5}
		\node[shape=rectangle, fill=red, minimum size=4pt] (\j8) at (\i,7.5) {};
		\node[shape=rectangle, fill=red, minimum size=4pt] (\j-8) at (\i,-7.5) {};
	}
	
	\foreach \i in {2,4,6,8} 
	{
		\pgfmathtruncatemacro\j{\i+1}
		\draw (\i7)--(\i8)--(\i9)--(\j8)--(\i7);
		\draw (\i-7)--(\i-8)--(\i-9)--(\j-8)--(\i-7);
	}
	
	\end{tikzpicture}
	\label{fig:skeleton}
	}
	\subfloat[]{
	\centering
	\begin{tikzpicture}[scale=0.20]
	
	\draw[ultra thick] (-2,8.5) -- (22, 8.5); 
	\node[draw=none] at (-3,8.5) {$\ell_\Delta$};
	\draw[ultra thick] (-2,0) -- (22,0); 
	\node[draw=none] at (-3,0) {$\ell_\alpha$};
	\draw[ultra thick] (-2,-8.5) -- (22, -8.5); 
	\node[draw=none] at (-3,-8.5) {$\ell_\nabla$};

	\draw[ultra thick] (0,-10) -- (0, 10); 
	\node[draw=none] at (0,-11) {$\ell_L$};
	\draw[ultra thick] (20,-10) -- (20, 10); 
	\node[draw=none] at (20,-11) {$\ell_R$};
	
	\node[draw=none] at (-3,2.5) {$\ell_A$};
	\node[draw=none] at (-3,4.5) {$\ell_B$};
	\node[draw=none] at (-3,6.5) {$\ell_C$};
	\node[draw=none] at (-3,-2.5) {$\ell_A'$};
	\node[draw=none] at (-3,-4.5) {$\ell_B'$};
	\node[draw=none] at (-3,-6.5) {$\ell_C'$};
	\node[draw=none] at (4,-11) {$\ell_q$};
	\node[draw=none] at (8,-11) {$\ell_r$};
	\node[draw=none] at (12,-11) {$\ell_s$};
	\node[draw=none] at (16,-11) {$\ell_t$};
 
	\foreach \i in {1,2,3} 
	{
		\draw (-2,2*\i+0.5) -- (22,2*\i+0.5);
		\draw (-2,2*-\i-0.5) -- (22,2*-\i-0.5);
		
	}

	\foreach \i in {1,2,3,4} 
	{
		\draw (4*\i,-10) -- (4*\i, 10);
		
	}
	\end{tikzpicture}
\label{fig:frame}
}
	
	\subfloat[]{
	\centering
	
	\begin{tikzpicture}[scale=0.3]

\node[draw=none] at (-3.5,0) {$\ell_\alpha$};
\node[draw=none] at (0,-12) {$\ell_L$};
\node[draw=none] at (19,-12) {$\ell_R$};

\node[draw=none] at (3.8,-12) {$\ell_p$};
\node[draw=none] at (7.6,-12) {$\ell_q$};
\node[draw=none] at (11.4,-12) {$\ell_r$};
\node[draw=none] at (15.2,-12) {$\ell_s$};

\node[draw=none] at (-3.5,-9) {$\ell_\nabla$};
\node[draw=none] at (-3.5,-7.05) {$\ell'_C$};
\node[draw=none] at (-3.5,-4.85) {$\ell'_B$};
\node[draw=none] at (-3.5,-2.65) {$\ell'_A$};
\node[draw=none] at (-3.5,2.65) {$\ell_A$};
\node[draw=none] at (-3.5,4.85) {$\ell_B$};
\node[draw=none] at (-3.5,7.05) {$\ell_C$};
\node[draw=none] at (-3.5,9) {$\ell_\Delta$};

\tikzstyle{every node}=[draw, fill=black, shape=circle, minimum size=2pt,inner sep=0pt];

\draw[thick] (-3,0)--(22,0);

\foreach \i in {0,1.9,3.8,5.7,7.6,9.5,11.4,13.3,15.2,17.1,19} 
{
	\node at (\i,0) {};
	\draw (\i,0) circle[radius=1];			
}

\foreach \i in {3.8,7.6,11.4,15.2} 
{
	\foreach \j in {-7.6,-6.5,-5.4,-4.3,-3.2,-2.1,-1.05,1.05,2.1,3.3, 4.4,5.5, 6.6,7.7}
	{
		\node at (\i,\j) {};
		\draw (\i,\j) circle[radius=1];
		
	}
}

	\foreach \i in {0,19} 
	{
		\foreach \j in {-7.6,-6.5,-5.4,-4.3,-3.2,-1.9,1.9,3.3, 4.4,5.5, 6.6,7.7}
		{
			\node at (\i,\j) {};
			\draw (\i,\j) circle[radius=1];
			
		}
	}

\foreach \i in {0,3.8,7.6,11.4,15.2,19} 
\draw (\i,-11.5)--(\i,11.5);	

\foreach \i in {-7.05,-4.85,-2.65,2.65,4.85,7.05} 
{	
	\node at (1.5,\i) {};
	\node at (-1.5,\i) {};
	\node at (17.5,\i) {};
	\node at (20.5,\i) {};		
	\draw (1.5,\i) circle[radius=1];
	\draw (-1.5,\i) circle[radius=1];
	\draw (17.5,\i) circle[radius=1];
	\draw (20.5,\i) circle[radius=1];
	\draw (-3,\i)--(22,\i);
}

\draw[thick] (-3,9)--(22,9); 
\draw[thick] (-3,-9)--(22,-9); 

\foreach \i in {0,2.32, 4.38, 6.43, 8.48, 10.53, 12.58, 14.63, 16.68,19}
{
	\node at (\i,9) {};
	\draw (\i,9) circle[radius=1];
	\node at (\i,-9) {};
	\draw (\i,-9) circle[radius=1];
}

\foreach \i in {3.8,7.6,11.4,15.2} 
{
	\foreach \j in {-9.8,9.8}
	{
		\node at (\i,\j) {};
		\draw (\i,\j) circle[radius=1];
	}
}

\end{tikzpicture}
\label{fig:skeletonRealization}
}
\caption{(a) Skeleton of the input graph for $\Phi$. The consecutive induced paths, labeled as $P_q, P_r, P_s, P_t$, are to be embedded on the literal lines $\ell_q, \ell_r, \ell_s, \ell_t$ in Figure~\ref{fig:frame} respectively. The vertices in the long induced paths $P_L$ and $P_R$ in \ref{fig:skeleton} (indicated by rectangles) must be embedded on the lines $\ell_L$ and $\ell_R$ given in \ref{fig:frame}. Similarly, the vertices in $P_\alpha$ (indicated by blue and green rectangles) must be embedded on the line $\ell_\alpha$ given in \ref{fig:frame}.\\
(b) The line set of the configuration for a Monotone NAE3SAT formula $\Phi$ with $4$ literals ($q,r,s,t$) and $3$ clauses ($A,B,C$).\\
(c) Realization of the graph given in \ref{fig:skeleton} onto the lines given in \ref{fig:frame}. }
	\label{fig:frameSkeletonRealization}
\end{figure}

Three essential components of the input graph is the following induced paths
$P_\alpha = (\alpha_1, \alpha_2, \dots, \alpha_{11})$, 
$P_L = (L_1, L_2, \dots, L_{15})$, 
and $P_R = (R_1, R_2, \dots, R_{15})$.
The length of $P_\alpha$ is $2m+3$ for $m$ literals. In our case, $(2\times 4) +3 = 11$.
The lengths of $P_L$ and $P_R$ are the same, equal to $3 + 4k$ for $k$ clauses. In our case, $3 + (4\times 3) = 15$.

The middle vertices of $P_L$ and $P_R$ are the end vertices of $P_\alpha$.
That is, $\alpha_1 = L_8$, and $\alpha_{11} = R_8$.
The paths $P_L$ and $P_R$ define the left and the right boundary for our gadget, respectively.

For $i  = q,r,s,t$, there is an induced path $P_i = (i_1, \dots, i_{15})$ for each literal, with 15 vertices.
In general, we denote those paths by $P^1, P^2, \dots, P^m$ for $m$ literals.
The vertices of these paths are denoted by blue circles in the Figure~\ref{fig:skeleton}, they are mutually disjoint, but each of them shares one vertex with $P_\alpha$.
The shared vertices are precisely the middle vertices, those are indicated by green rectangle vertices in the figure.
That is, $\alpha_3 = q_8$, $\alpha_5 = r_8$, $\alpha_7 = s_8$, and $\alpha_9 = t_8$.
Moreover, $i_1$ is a vertex of an induced 4-cycle, and $i_{15}$ is a vertex of another induced 4-cycle for $i = q,r,s,t$.
The three vertices in a 4-cycle, except the one in one of the induced paths, are indicated by the red color in the figure.
Precisely two of them, that are adjacent to a blue vertex (either $i_1$ or $i_{15}$) are indicated by squares, and the remaining is indicated by a triangle.

Starting from the second edge of $P_L$ (respectively $P_R$), every second edge is a chord of a 4-cycle ($C_4$).
Throughout the paper, we refer to such 4-cycles with a chord as a \emph{diamond}.
Two vertices of these diamonds are of $P_L$ (respectively $P_R$), and remaining two are denoted by red triangles in Figure~\ref{fig:skeleton}.

Remember that the problem takes two inputs: a graph, and a set of lines determined by their equations (or rather by two sets of rational numbers, since every line is parallel to either $x$- or $y$- axis).
For a Monotone NAE3SAT formula with $3$ clauses and $4$ literals, we have described the input graph above.
Now, let us discuss the input lines of our gadget.
The input graph is given in Figure~\ref{fig:skeleton}, and the corresponding lines are given in Figure~\ref{fig:frame}.
We claim that the given graph can be embedded onto the given lines with $\varepsilon$ flexibility, and the resulting realization looks like the set of unit disks given in Figure~\ref{fig:skeletonRealization}.

In order to force such embedding, we adjust the Euclidean distance between each pair of parallel lines carefully.
We start by defining the horizontal line $\ell_\alpha$.
This line is the axis of horizontal symmetry for our line configuration.
Thus, it is safe to assume that $\ell_\alpha$ is the $x$-axis.
On the positive side of the $y$-axis, for each clause $A$, $B$, and $C$, there is a straight line parallel to $\ell_\alpha$, and another horizontal line acting as the top boundary of the configuration.
These lines are denoted by $\ell_A, \ell_B$, $\ell_C$, and $\ell_\Delta$, and their equations are $y = a$, $y = b$, $y = c$, and $y = \Delta$, respectively, where $a < b < c < \Delta$.
For every pair of consecutive horizontal lines, the Euclidean distance between them is precisely $2.01$ units.
That is, $a = 2.01$, $b = 4.02$ and $c = 6.03$, and $\Delta = 8.04$.
For every horizontal line described above, there is another horizontal line symmetric to it about the $x$-axis.
These lines are $\ell'_A, \ell'_B, \ell'_C$, and $\ell_\nabla$ (see Figure~\ref{fig:frame}).

The leftmost vertical line is $\ell_L$, which is the left boundary of our configuration.
We can safely assume that $\ell_L$ is the $y$-axis for the sake of simplicity.
For each literal $q$, $r$, $s$ and $t$, there exists a vertical line parallel to $\ell_L$, and another vertical line that defines the right boundary of our configuration.
These lines are denoted by $\ell_q, \ell_r$, $\ell_s$, $\ell_t$, and $\ell_R$, and their equations are $x = q$, $x = r$, $x = s$, $x = t$ and $x = R$, respectively, where $q < r < s < t < R$.
The Euclidean distance between each pair of consecutive vertical lines is precisely $3.8$ units.
That is $q = 3.8$, $r = 7.6$, $s = 11.4$, $t = 15.2$, and $R = 19$.

Up to this point, we have described the input graph, and the input lines for a given Monotone NAE3SAT formula with 3 clauses and 4 literals.
In general, for a given Monotone NAE3SAT formula $\Phi$ with $k$ clauses $C_1, C_2, \dots, C_k$, and $m$ literals $x_1, x_2, \dots, x_m$, our gadget has the following components.
\begin{enumerate}
\item An induced path $P_\alpha = (\alpha_1, \alpha_2, \dots, \alpha_{2m+3})$ with $2m+3$ vertices.
\item $m$ induced paths $P^1 = (P^1_1, P^1_2, \dots, P^1_{4k+3})$, $\dots$, $P^m = (P^m_1, P^m_2, \dots, P^m_{4k+3})$, each with $4k+3$ vertices, where $\alpha_3 = P^1_{2k+2}$, $\alpha_5 = P^2_{2k+2}$, $\dots$, $\alpha_{2k+1} = P^m_{2k+2}$, and induced 4-cycles containing the first and the last vertices of each of these paths.
\item Two induced paths $P_L = (L_1, \dots, L_{4k+3})$ and $P_R = (R_1, \dots, R_{4k+3})$, each with $4k+3$ vertices, where the edges $L_2L_3$, $L_4L_5$, $\dots$, $L_{2k}L_{2k+1}$, $L_{2k+3}L_{2k+4}$, $\dots$, $L_{4k+1}L_{4k+2}$, and $R_2R_3$, $R_4R_5$, $\dots$, $R_{2k}R_{2k+1}$, $R_{2k+3}R_{2k+4}$, $\dots$,\\ $R_{4k+1}R_{4k+2}$ are chords of disjoint 4-cycles.
\item $2k+3$ horizontal lines $\ell_\nabla, \ell'^C_k, \ell'^C_{k-1}, \dots, \ell_\alpha, \ell^C_1, \ell^C_2, \dots, \ell^C_k, \ell_\Delta$, with equations $\ell_\nabla: y= -2.01(k+1)$, $\ell_\Delta: -\ell_\nabla$, $\ell_\alpha: y = 0$, $\ell'^C_{i} = -2.01i$, and $\ell^C_{i} = 2.01i$ for $i = 1,2,\dots,k$.
\item $m+2$ vertical lines $\ell_L, \ell^x_1, \ell^x_2, \dots, \ell^x_m, \ell_R$, with equations $\ell_L: x = 0$, $\ell_R: x = 3.8(m+1)$, and $\ell^x_i: 3.8i$ for $i = 1,2,\dots,m$.
\end{enumerate}
In total, for the given formula $\Phi$ with $k$ clauses and $m$ literals, our gadget is an instance of $\APUD(2k+3,m+2)$.
Here, we conclude the proof of Theorem~\ref{thm:main}.

Now, let us show that the given graph has a unique embedding onto the given lines, up to $\varepsilon$ flexibility.

\begin{claim}
The vertices indicated by rectangles in Figure~\ref{fig:skeleton} can only be embedded on the bold lines in Figure~\ref{fig:frame}.
\end{claim}

Let us start by discussing the embedding of $P_L$ onto $\ell_L$ (and respectively $P_R$ onto $\ell_R$).
We give the following two trivial lemmas as preliminaries for the proof of our claim.

\begin{lemma} \label{lem:triangle}
	Consider two disks $A$ and $B$, centered on $(a,0)$ and $(b,0)$ with $0<|a|<|b|$.
	Another disk, $C$ that is centered on $(0,c)$ cannot intersect $B$ without intersecting $A$.
\end{lemma}

\begin{proof}

Consider the triangle whose corners are $(a,0)$, $(b,0)$, and $(0,c)$.
If $|a|<|b|$, then $\sqrt{a^2+c^2} < \sqrt{b^2+c^2}$ holds.
For $C$ to intersect $B$, $\sqrt{b^2+c^2} \leq 2$ must hold.
However, since $|a| < |b|$, if $\sqrt{b^2+c^2} \leq 2$ holds, then $\sqrt{a^2+c^2} \leq 2$ also holds.
Thus, $C$ intersects $B$ if, and only if $C$ intersects $A$. 
\end{proof}

\begin{lemma} \label{lem:s4}
	An induced 4-star ($K_{1,4}$) can be realized as a unit disk graph on two perpendicular lines, but not on two parallel lines.
\end{lemma}

\begin{proof}

First part of the proof is trivial.
Consider the $x=0$ and $y=0$ lines as two perpendicular lines.
Four unit disks centered on $(0,-(2-\varepsilon)$, $((2-\varepsilon),0)$, $(0, (2-\varepsilon)$, and $(-(2-\varepsilon),0)$ where $0 < \varepsilon \ll 1$ form an induced 4-star ($K_{1,4}$).

Now, let us show that an induced 4-star cannot be realized as unit disks on two parallel lines.
The disks that correspond to the vertices of an induced claw ($K_{1,3}$) can be embedded on two parallel lines if the centers of three disks are collinear.
Suppose that four disks, $a$, $b$, $c$, and $u$ form an induced claw, where $u$ is the central vertex, and $a,b,c$ are the rays.
Without loss of generality, suppose that $a$, $u$, and $c$ are centered on $(A,0)$, $(0,0)$ and $(C,0)$, respectively.
Thus, $b$ must be on the second parallel line, say $y=k$.
To complete a 4-star, we need one more disk, $d$, centered on either $y=0$ or $y=k$, such that $d$ intersects $u$, but none of $a$, $b$, and $c$ 

Clearly, $d$ cannot be on $y=0$ line because $u$ is enclosed by $a$ and $c$ from both sides.
So, suppose that $d$ is centered on $(D,k)$.
In this case, we show that no such $k$ exists by contradiction.

Place two more disks, $b'$ and $d'$, centered on $(-B,k)$ and $(-D,k)$, respectively.
If $b$ and $d$ do not intersect, then $b'$ and $d'$ also do not intersect.
Moreover, since $a$ does not intersect with $b$, $a$ also does not intersect with $b'$.
Symmetrically, $c$, $d$, $c'$, and $d'$ have no pairwise intersections.

The described configuration is a $K_{1,6}$ with vertices $u;a,b,c,d,b',d'$ which cannot be realized as a unit disk graph.
Therefore, we have a contradiction. 
\end{proof}

Now, with the help of Lemmas~\ref{lem:triangle} and \ref{lem:s4}, we state the following lemmas, and prove our claim.

\begin{lemma} \label{lem:LRa}
The induced paths $P_L$, $P_R$ and $P_\alpha$ in the input graph (Figure~\ref{fig:skeleton}) can only be embedded onto $\ell_L$, $\ell_R$, and $\ell_\alpha$, respectively (Figure~\ref{fig:frame}).
\end{lemma}

\begin{proof}
Assume that $P_L = (L_1, \dots, L_{15})$ is realized on a single line, say $d$.
Denote the disks $\mathcal{D}_L^1, \dots \mathcal{D}_L^{15}$ that correspond to these vertices, respectively.
Denote the disks $\mathcal{D}_\alpha^1, \dots, \mathcal{D}_\alpha^{11}$ that correspond to the vertices ($\alpha_1, \dots, \alpha_{11}$) of $P_\alpha$.

Consider the diamond $u,v,L_{11}, L_{12}$, where $u$ and $v$ are indicated by red triangles in Figure~\ref{fig:skeleton}.
Neither $U$ nor $V$ intersect any disk that corresponds to a vertex of $P_L$.
Therefore, neither of them can be embedded onto $d$.
Similarly, if they are embedded onto some other line perpendicular to $d$, then they intersect either $\mathcal{D}_L^{10}$ or $\mathcal{D}_L^{13}$ to intersect $\mathcal{D}_L^{11}$ and $\mathcal{D}_L^{12}$.
Thus, every such pair of disks must be embedded in a way such that their centers lie onto the same line that intersects $d$, and at the different sides of $d$.
This property also holds for the diamond that includes the vertices $L_{13}$ and $L_{14}$, as there is one extra vertex $L_{15}$ at the end of $P_L$.

There are 6 disjoint diamonds in which 12 vertices of $P_L$ are included.
Each of these diamonds must be embedded around the intersection of two perpendicular lines.
In addition, there is one vertex that $P_L$ and $P_\alpha$ have in common ($\alpha_1 = L_8$).
Since $P_\alpha$ is an induced path, each disk $\mathcal{D}_\alpha^i$ such that $1 \leq i \leq 11$ cannot be centered on $d$.
Thus, the center of the disk $\mathcal{D}_L^8 = \mathcal{D}_\alpha^1$, must be the closest center to some line $f$ on which $\mathcal{D}_\alpha^2$ is centered.
Otherwise, due to Lemma~\ref{lem:triangle}, $\mathcal{D}_\alpha^2$ intersects with a closer to $f$.

Therefore, to embed $P_L$, there must be a straight line, and precisely 7 parallel lines that are perpendicular to this line.
Similarly, $P_R$ also requires 7 intersection points to be realized.
Thus, both $P_L$ and $P_R$ must be realized on a vertical line.

Now recall Lemma~\ref{lem:s4}.
Note that each one of $\alpha_3, \alpha_5, \alpha_7$, and $\alpha_9$ is the central vertex of some 4-star, and thus must be embedded near the intersection of two perpendicular lines.
Considering that $P_L$ and $P_R$ are realized on two vertical lines, we need 6 intersection points to realize $P_\alpha$ as unit disks.
Since $\alpha_1 = L_8$ and $\alpha_{11} = R_8$, the remaining four intersections for 4-stars are the intersections between $\ell_alpha$ and each of $\ell_q, \ell_r, \ell_s, \ell_t$.
It is now trivial to see that our initial assumption of $P_L$ being realized on a single line always holds, since if $P_L$ is realized on more than one straight lines, then $P_\alpha$ cannot be realized due to the number of intersections.
Of course, one also needs to take into account the other induced paths, $P_q, P_r, P_s, P_t$ to verify this statement.

Therefore, $P_L$, $P_R$ and $P_\alpha$ can only be embedded onto $\ell_L$, $\ell_R$, and $\ell_\alpha$, respectively. 
\end{proof}

\begin{claim} \label{lem:3clauses4literals}
For the given input graph for 3 clauses and 4 literals, the following hold:
\begin{enumerate}[i)]
\item The induced paths $P_q = (q_1, \dots, q_{15})$, $P_r = (r_1, \dots, r_{15})$, $P_s = (s_1, \dots, s_{15})$ and $P_t = (t_1, \dots, t_{15})$ in the input graph given in Figure~\ref{fig:skeleton} can only be embedded onto $\ell_q$, $\ell_r$, 
$\ell_s$, and $\ell_t$, respectively.
\item The center of each disk that correspond to a vertex of those induced paths must be between $\ell_\Delta$ and $\ell_\nabla$. 
\item A pair of non-intersecting disks that are included in an induced 4-cycle, but not included in any of $P_q, P_r, P_s, P_t$, must lie on either $\ell_\Delta$ or $\ell_\nabla$ (red rectangles in Figure~\ref{fig:skeleton}). 
\end{enumerate}
\end{claim}

\begin{lemma} \label{lem:KclausesMliterals}
For the given input graph for $k$ clauses and $m$ literals, the following hold:
\begin{enumerate}[i)]
\item The induced paths $P_1 = (P^1_1, \dots, P^1_{4k+3})$, $P^2 = (P^2_1, \dots, P^2_{4k+3})$, $\dots$, $P^n = (P^n_1, \dots, P^n_{4k+3})$ in the input graph can only be embedded onto $\ell^x_1$, $\ell^x_2$, $\dots$ $\ell^x_m$, respectively.
\item The center of each disk that correspond to a vertex of those induced paths must be between $\ell_\Delta$ and $\ell_\nabla$. 
\item A pair of non-intersecting disks that are included in an induced 4-cycle, but not included in any of $P^1, P^2, \dots, P^n$, must lie on either $\ell_T$ or $\ell_B$ (red rectangles in Figure~\ref{fig:skeleton}). 
\end{enumerate}
\end{lemma}

\begin{proof}
Due to Lemma~\ref{lem:LRa}, we know that $P_\alpha$ is realized on $\ell_\alpha$.
Also recall that $\alpha_3 = P^1_{2k+2}$, $\alpha_5 = P^2_{2k+2}$, $\dots$, $\alpha_{2k+1} = P^m_{2k+2}$.
Thus, $P^1, P^2, \dots, P^n$ must be realized on $\ell^x_1, \ell^x_2, dots, \ell^x_m$, respectively.

Notice the induced 4-cycles in Figure~\ref{fig:skeleton}.
These 4-cycles can be embedded only on two different straight lines, since a 4-cycle is a forbidden subgraph in a unit interval graph.
The Euclidean distance between a pair of consecutive vertical lines is greater than 2 units.
Therefore, the disks that correspond to those 4-cycles must be centered on two perpendicular lines.

Consider $2m+1$ disks that correspond to $P_i^{2m+2}, \dots, P_i^{4m+3}$.
To realize an induced path of length $2m+1$ as unit disks on a straight line, the distance between the disks at two ends of the path must be greater than  units.
However, if the centers of the disks are between $\ell_\alpha$ ($x=0$) and $\ell_C^m$ ($x=3.8m$), then two disks of the 4-cycle must be on $\ell_C$, which is an infeasible configuration, because the center of the disk that corresponds to $q_15$ has $y$-coordinate greater than $6$.
Thus, topmost center among the centers of the mentioned 8 disks must above $\ell_C$, and below $\ell_T$.

In any 4-cycle, two of the disks must be centered on the same line, and the remaining two must be centered a line which is perpendicular to the line on which the first two are centered.
Moreover, the centers must lie at four different directions from the intersection point.

Therefore, a non-intersecting pair of disks from each induced 4-cycle must be centered on $\ell_T$ and $\ell_B$.
The vertices to which those pairs of disks correspond are indicated by red rectangles in Figure~\ref{fig:skeleton}. 
\end{proof}

With the Lemmas~\ref{lem:LRa} and \ref{lem:KclausesMliterals}, we have shown that the vertices denoted by rectangles in Figure~\ref{fig:skeleton} must be embedded onto the bold lines in Figure~\ref{fig:frame}.

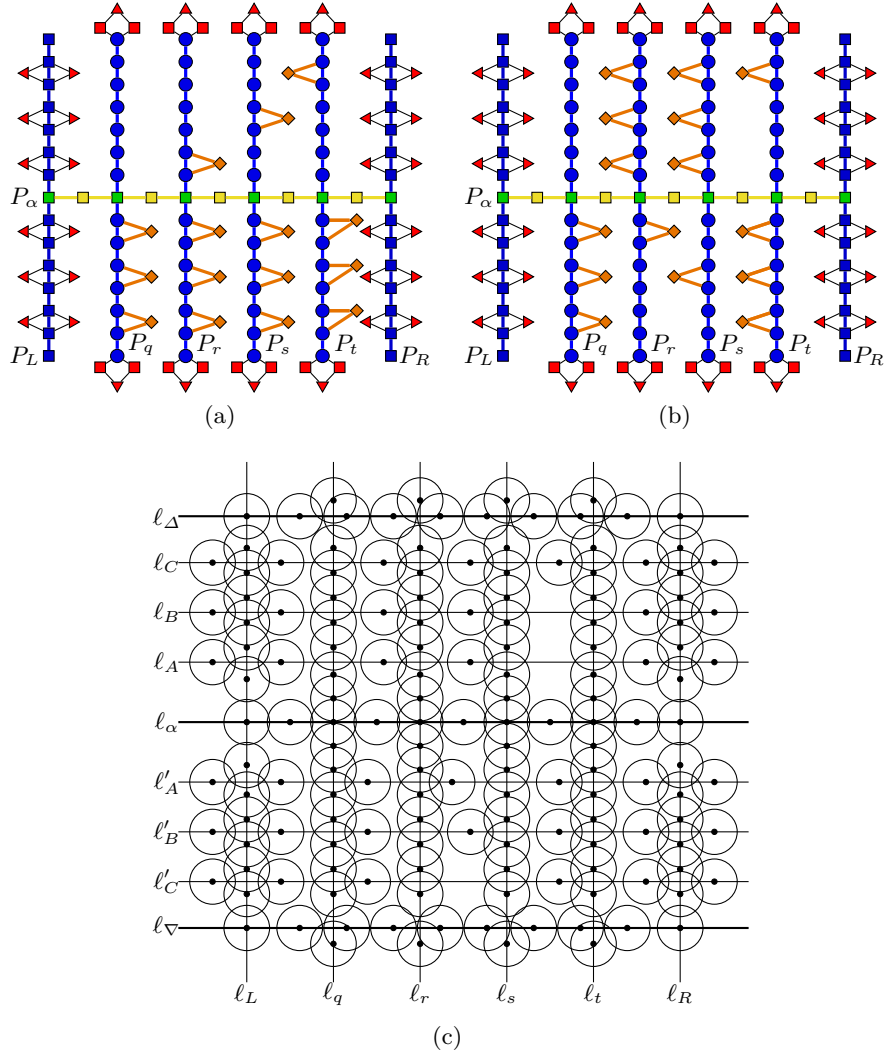
\begin{figure}[htbp]
	\centering
	
	\subfloat[]{
	\centering

	\begin{tikzpicture}[xscale=0.45, yscale=0.3,
	triangle/.style = {regular polygon, regular polygon sides=3},
	ltriangle/.style = {regular polygon, regular polygon sides=3, rotate=90},
	rtriangle/.style = {regular polygon, regular polygon sides=3, rotate=270},
	dtriangle/.style = {regular polygon, regular polygon sides=3, rotate=180}]
	
	\node[draw=none] at (-0.7,-7) {$P_L$};
	\node[draw=none] at (10.7,-7) {$P_R$};
	\node[draw=none] at (-0.7,0) {$P_\alpha$};
	\node[draw=none] at (2.7,-6.5) {$P_q$};
	\node[draw=none] at (4.7,-6.5) {$P_r$};
	\node[draw=none] at (6.7,-6.5) {$P_s$};
	\node[draw=none] at (8.7,-6.5) {$P_t$};

	\tikzstyle{every node}=[draw, shape=circle, minimum size=5pt,inner sep=0pt];

	\foreach \i in {0,2,4,6,8,10} \node[shape=rectangle, fill=green!80!black, minimum size=4pt] (\i0) at (\i,0) {};
	
	\foreach \i in {1,3,5,7,9}	\node[shape=rectangle, fill=yellow!90!black, minimum size=4pt] (\i0) at (\i,0) {};
	
	\foreach \i in {0,...,9}
	{
		\pgfmathtruncatemacro\j{\i+1};
		\draw[very thick, color=yellow!90!black] (\i0)--(\j0);
	}
		
	\foreach \i in {2,4,6,8} 
	{
		\foreach \j in {-7,-6,-5,-4,-3,-2,-1,1,2,3,4,5,6,7}
		{
			\node[fill=blue!90!black] (\i\j) at (\i,\j) {};
			
		}
	}
	
	\foreach \i in {2,4,6,8} 
	{
		\foreach \j in {-7,-6,-5,-4,-3,-2,-1,0,1,2,3,4,5,6}
		{
			\pgfmathtruncatemacro\k{\j+1};
			\draw[very thick, color=blue] (\i\j)--(\i\k);
			
		}
	}

	\foreach \i in {0,10} 
	{
		\foreach \j in {-7,-6,-5,-4,-3,-2,-1,1,2,3,4,5,6,7}
		{
			\node[shape=rectangle,fill=blue!80!black, minimum size=4pt] (\i\j) at (\i,\j) {};
			
		}
	}

	\foreach \i in {0,10} 
	{
		\foreach \j in {-7,-6,-5,-4,-3,-2,-1,0,1,2,3,4,5,6}
		{
			\pgfmathtruncatemacro\k{\j+1};
			\draw[very thick, color=blue!80!black] (\i\j)--(\i\k);
			
		}
	}

	\foreach \i in {-5.5,-3.5,-1.5,1.5,3.5,5.5}
	{
		\pgfmathtruncatemacro\j{\i-0.5}; 
		\pgfmathtruncatemacro\k{\i+0.5};

		\node[ltriangle, fill=red] (-1\j) at (-0.7,\i) {};		
		\node[rtriangle, fill=red] (1\j) at (0.7,\i) {};
		
		\node[ltriangle, fill=red] (9\j) at (9.3,\i) {};
		\node[rtriangle, fill=red] (11\j) at (10.7,\i) {};
		\draw (1\j)--(0\j);
		\draw (1\j)--(0\k);
		\draw (-1\j)--(0\j);
		\draw (-1\j)--(0\k);
		\draw (9\j)--(10\j);
		\draw (9\j)--(10\k);
		\draw (11\j)--(10\j);
		\draw (11\j)--(10\k);
	}

	\foreach \i in {2,4,6,8} 
	{
		\node[triangle, fill=red] (\i9) at (\i,8.3) {};
		\node[dtriangle, fill=red] (\i-9) at (\i,-8.3) {};
	}
	
	\foreach \i in {1.5,2.5,3.5,4.5,5.5,6.5,7.5,8.5} 
	{
		\pgfmathtruncatemacro\j{\i+0.5}
		\node[shape=rectangle, fill=red, minimum size=4pt] (\j8) at (\i,7.5) {};
		\node[shape=rectangle, fill=red, minimum size=4pt] (\j-8) at (\i,-7.5) {};
	}
	
	\foreach \i in {2,4,6,8} 
	{
		\pgfmathtruncatemacro\j{\i+1}
		\draw (\i7)--(\i8)--(\i9)--(\j8)--(\i7);
		\draw (\i-7)--(\i-8)--(\i-9)--(\j-8)--(\i-7);
	}
	
		\foreach \i in {-5.5,-3.5,-1.5}
	{
		\pgfmathtruncatemacro\j{\i-0.5}; 
		\pgfmathtruncatemacro\k{\i+0.5};
		\pgfmathtruncatemacro\l{\k-1}; 
		\foreach \x in {3,5,7}
		{
			\pgfmathtruncatemacro\y{\x-1}; 
			\pgfmathtruncatemacro\z{\j+1};
			\node[shape=diamond,fill=orange!90!black] (\x\j) at (\x,\i) {};
			\draw[very thick, color=orange!90!black] (\x\j)--(\y\j);
			\draw[very thick, color=orange!90!black] (\x\j)--(\y\z);	
		}
		\node[shape=diamond,fill=orange!90!black] (9\k) at (9,\k) {};
		\draw[very thick, color=orange!90!black] (9\k)--(8\k);
		\draw[very thick, color=orange!90!black] (9\k)--(8\k);
		\draw[very thick, color=orange!90!black] (9\k)--(8\l);
		
	}

	\node[shape=diamond, fill=orange!90!black] (52) at (5, 1.5) {};
	\draw[very thick, color=orange!90!black] (52)--(41);
	\draw[very thick, color=orange!90!black] (52)--(42);
	
	\node[shape=diamond,fill=orange!90!black] (74) at (7, 3.5) {};
	\draw[very thick, color=orange!90!black] (74)--(63);
	\draw[very thick, color=orange!90!black] (74)--(64);
	
	\node[shape=diamond,fill=orange!90!black] (75) at (7,5.5) {};
	\draw[very thick, color=orange!90!black] (75)--(85);
	\draw[very thick, color=orange!90!black] (75)--(86);
	
	\end{tikzpicture}
	\label{fig:withFlags}
	}
	\subfloat[]{
	\centering
	\begin{tikzpicture}[xscale=0.45, yscale=0.3,
	triangle/.style = {regular polygon, regular polygon sides=3},
	ltriangle/.style = {regular polygon, regular polygon sides=3, rotate=90},
	rtriangle/.style = {regular polygon, regular polygon sides=3, rotate=270},
	dtriangle/.style = {regular polygon, regular polygon sides=3, rotate=180}]
	
	\node[draw=none] at (-0.7,-7) {$P_L$};
	\node[draw=none] at (10.7,-7) {$P_R$};
	\node[draw=none] at (-0.7,0) {$P_\alpha$};
	\node[draw=none] at (2.7,-6.5) {$P_q$};
	\node[draw=none] at (4.7,-6.5) {$P_r$};
	\node[draw=none] at (6.7,-6.5) {$P_s$};
	\node[draw=none] at (8.7,-6.5) {$P_t$};

	\tikzstyle{every node}=[draw, shape=circle, minimum size=5pt,inner sep=0pt];

	\foreach \i in {0,2,4,6,8,10} \node[shape=rectangle, fill=green!80!black, minimum size=4pt] (\i0) at (\i,0) {};
	
	\foreach \i in {1,3,5,7,9}	\node[shape=rectangle, fill=yellow!90!black, minimum size=4pt] (\i0) at (\i,0) {};
	
	\foreach \i in {0,...,9}
	{
		\pgfmathtruncatemacro\j{\i+1};
		\draw[very thick, color=yellow!90!black] (\i0)--(\j0);
	}
		
	\foreach \i in {2,4,6,8} 
	{
		\foreach \j in {-7,-6,-5,-4,-3,-2,-1,1,2,3,4,5,6,7}
		{
			\node[fill=blue!90!black] (\i\j) at (\i,\j) {};
			
		}
	}
	
	\foreach \i in {2,4,6,8} 
	{
		\foreach \j in {-7,-6,-5,-4,-3,-2,-1,0,1,2,3,4,5,6}
		{
			\pgfmathtruncatemacro\k{\j+1};
			\draw[very thick, color=blue] (\i\j)--(\i\k);
			
		}
	}

	\foreach \i in {0,10} 
	{
		\foreach \j in {-7,-6,-5,-4,-3,-2,-1,1,2,3,4,5,6,7}
		{
			\node[shape=rectangle,fill=blue!80!black, minimum size=4pt] (\i\j) at (\i,\j) {};
			
		}
	}

	\foreach \i in {0,10} 
	{
		\foreach \j in {-7,-6,-5,-4,-3,-2,-1,0,1,2,3,4,5,6}
		{
			\pgfmathtruncatemacro\k{\j+1};
			\draw[very thick, color=blue!80!black] (\i\j)--(\i\k);
			
		}
	}

	\foreach \i in {-5.5,-3.5,-1.5,1.5,3.5,5.5}
	{
		\pgfmathtruncatemacro\j{\i-0.5}; 
		\pgfmathtruncatemacro\k{\i+0.5};

		\node[ltriangle, fill=red] (-1\j) at (-0.7,\i) {};		
		\node[rtriangle, fill=red] (1\j) at (0.7,\i) {};
		
		\node[ltriangle, fill=red] (9\j) at (9.3,\i) {};
		\node[rtriangle, fill=red] (11\j) at (10.7,\i) {};
		\draw (1\j)--(0\j);
		\draw (1\j)--(0\k);
		\draw (-1\j)--(0\j);
		\draw (-1\j)--(0\k);
		\draw (9\j)--(10\j);
		\draw (9\j)--(10\k);
		\draw (11\j)--(10\j);
		\draw (11\j)--(10\k);
	}

	\foreach \i in {2,4,6,8} 
	{
		\node[triangle, fill=red] (\i9) at (\i,8.3) {};
		\node[dtriangle, fill=red] (\i-9) at (\i,-8.3) {};
	}
	
	\foreach \i in {1.5,2.5,3.5,4.5,5.5,6.5,7.5,8.5} 
	{
		\pgfmathtruncatemacro\j{\i+0.5}
		\node[shape=rectangle, fill=red, minimum size=4pt] (\j8) at (\i,7.5) {};
		\node[shape=rectangle, fill=red, minimum size=4pt] (\j-8) at (\i,-7.5) {};
	}
	
	\foreach \i in {2,4,6,8} 
	{
		\pgfmathtruncatemacro\j{\i+1}
		\draw (\i7)--(\i8)--(\i9)--(\j8)--(\i7);
		\draw (\i-7)--(\i-8)--(\i-9)--(\j-8)--(\i-7);
	}
	
	\foreach \i in {-5.5,-3.5,-1.5}
	{
		\pgfmathtruncatemacro\j{\i-0.5}; 
		\pgfmathtruncatemacro\k{\i+0.5};
		\pgfmathtruncatemacro\l{\k-1}; 
		\foreach \x in {3}
		{
			\pgfmathtruncatemacro\y{\x-1}; 
			\pgfmathtruncatemacro\z{\j+1};
			\node[shape=diamond,fill=orange!90!black] (\x\j) at (\x,\i) {};
			\draw[very thick, color=orange!90!black] (\x\j)--(\y\j);
			\draw[very thick, color=orange!90!black] (\x\j)--(\y\z);	
		}
		\node[shape=diamond,fill=orange!90!black] (9\k) at (7,\i) {};
		\draw[very thick, color=orange!90!black] (9\k)--(8\k);
		\draw[very thick, color=orange!90!black] (9\k)--(8\k);
		\draw[very thick, color=orange!90!black] (9\k)--(8\l);
		
	}
	
	\foreach \i in {5.5,3.5,1.5}
	{
		\pgfmathtruncatemacro\j{\i-0.5}; 
		\pgfmathtruncatemacro\k{\i+0.5};
		\pgfmathtruncatemacro\l{\k-1}; 
		\foreach \x in {3,5}
		{
			\pgfmathtruncatemacro\y{\x+1}; 
			\pgfmathtruncatemacro\z{\j+1};
			\node[shape=diamond,fill=orange!90!black] (\x\j) at (\x,\i) {};
			\draw[very thick, color=orange!90!black] (\x\j)--(\y\j);
			\draw[very thick, color=orange!90!black] (\x\j)--(\y\z);	
		}
		
	}

	\node[shape=diamond, fill=orange!90!black] (5-2) at (5, -1.5) {};
	\draw[very thick, color=orange!90!black] (5-2)--(4-1);
	\draw[very thick, color=orange!90!black] (5-2)--(4-2);
	
	\node[shape=diamond,fill=orange!90!black] (7-4) at (5, -3.5) {};
	\draw[very thick, color=orange!90!black] (7-4)--(6-3);
	\draw[very thick, color=orange!90!black] (7-4)--(6-4);
	
	\node[shape=diamond,fill=orange!90!black] (75) at (7,5.5) {};
	\draw[very thick, color=orange!90!black] (75)--(85);
	\draw[very thick, color=orange!90!black] (75)--(86);
	
	\end{tikzpicture}
\label{fig:truthAssignment}
}
	
	\subfloat[]{
	\centering
	
	\begin{tikzpicture}[scale=0.3]

\node[draw=none] at (-3.5,0) {$\ell_\alpha$};
\node[draw=none] at (0,-12) {$\ell_L$};
\node[draw=none] at (19,-12) {$\ell_R$};

\node[draw=none] at (3.8,-12) {$\ell_q$};
\node[draw=none] at (7.6,-12) {$\ell_r$};
\node[draw=none] at (11.4,-12) {$\ell_s$};
\node[draw=none] at (15.2,-12) {$\ell_t$};

\node[draw=none] at (-3.5,-9) {$\ell_\nabla$};
\node[draw=none] at (-3.5,-7.05) {$\ell'_C$};
\node[draw=none] at (-3.5,-4.85) {$\ell'_B$};
\node[draw=none] at (-3.5,-2.65) {$\ell'_A$};
\node[draw=none] at (-3.5,2.65) {$\ell_A$};
\node[draw=none] at (-3.5,4.85) {$\ell_B$};
\node[draw=none] at (-3.5,7.05) {$\ell_C$};
\node[draw=none] at (-3.5,9) {$\ell_\Delta$};

\tikzstyle{every node}=[draw, fill=black, shape=circle, minimum size=2pt,inner sep=0pt];

\draw[thick] (-3,0)--(22,0);

\foreach \i in {0,1.9,3.8,5.7,7.6,9.5,11.4,13.3,15.2,17.1,19} 
{
	\node at (\i,0) {};
	\draw (\i,0) circle[radius=1];			
}

\foreach \i in {3.8,7.6,11.4,15.2} 
{
	\foreach \j in {-7.6,-6.5,-5.4,-4.3,-3.2,-2.1,-1.05,1.05,2.1,3.3, 4.4,5.5, 6.6,7.7}
	{
		\node at (\i,\j) {};
		\draw (\i,\j) circle[radius=1];
		
	}
}

	\foreach \i in {0,19} 
	{
		\foreach \j in {-7.6,-6.5,-5.4,-4.3,-3.2,-1.9,1.9,3.3, 4.4,5.5, 6.6,7.7}
		{
			\node at (\i,\j) {};
			\draw (\i,\j) circle[radius=1];
			
		}
	}

\foreach \i in {0,3.8,7.6,11.4,15.2,19} 
\draw (\i,-11.5)--(\i,11.5);	

\foreach \i in {-7.05,-4.85,-2.65,2.65,4.85,7.05} 
{	
	\node at (1.5,\i) {};
	\node at (-1.5,\i) {};
	\node at (17.5,\i) {};
	\node at (20.5,\i) {};		
	\draw (1.5,\i) circle[radius=1];
	\draw (-1.5,\i) circle[radius=1];
	\draw (17.5,\i) circle[radius=1];
	\draw (20.5,\i) circle[radius=1];
	\draw (-3,\i)--(22,\i);
}

\draw[thick] (-3,9.1)--(22,9.1); 
\draw[thick] (-3,-9.1)--(22,-9.1); 

\foreach \i in {0,2.32, 4.38, 6.43, 8.48, 10.53, 12.58, 14.63, 16.68,19}
{
	\node at (\i,9.1) {};
	\draw (\i,9.1) circle[radius=1];
	\node at (\i,-9.1) {};
	\draw (\i,-9.1) circle[radius=1];
}

\foreach \i in {3.8,7.6,11.4,15.2} 
{
	\foreach \j in {-9.8,9.8}
	{
		\node at (\i,\j) {};
		\draw (\i,\j) circle[radius=1];
	}
}

\foreach \i in {2.65,4.85,7.05} 
{
	\node at (6,\i) {};
	\draw (6,\i) circle[radius=1];
	\node at (9.8,\i) {};
	\draw (9.8,\i) circle[radius=1];
}

\foreach \i in {-2.65,-4.85,-7.05} 
{
	\node at (5.3,\i) {};
	\draw (5.3,\i) circle[radius=1];
	\node at (13.7,\i) {};
	\draw (13.7,\i) circle[radius=1];
}

\node at (13.7,7.05) {};
\draw (13.7,7.05) circle[radius=1];

\node at (9.8,-4.85) {};
\draw (9.8,-4.85) circle[radius=1];

\node at (9,-2.65) {};
\draw (9,-2.65) circle[radius=1];
\end{tikzpicture}
\label{fig:truthAssignmentRealization}
}
\caption{(a) The input graph for the Monotone NAE3SAT formula with $\Phi$ with literals $q,r,s,t$, and clauses $A,B,C$.
$\Phi = A \wedge B \wedge C$  where $A = (q \vee s \vee t)$, $B = (q \vee r \vee t)$, and $C =(q \vee r \vee s)$. The flag vertices, indicated by orange diamonds, are adjacent to the vertices that correspond to a clause on an induced path, if the literal does not appear in that clause.\\
(b) A truth assignment that satisfies the formula given in \ref{fig:withFlags}: $q = \textsc{true}$, $r = \textsc{false}$, $s = \textsc{false}$, and $t = \textsc{true}$.\\
(c) Realization of the graph given in \ref{fig:truthAssignment}.}
	\label{fig:fullGadget}
\end{figure}

Using the backbone we have described, we now show how to model the relationship between the clauses and the literals.
To make it easier to follow, we also refer to Figure~\ref{fig:skeleton} in parentheses in the following description.
\begin{itemize}
\item Consider a sub-path $(P^i_{2k+3}, P^i_{2k+4}$, $\dots$, $P^i_{4k+3})$ of the induced path $P^i$. This part corresponds to the literal $x_i$ of the given Monotone NAE3SAT formula (corresponds to $(q_9, \dots, q_{15})$ of $P_q$ in our example).

\item The edges $P^i_{2k+3}P^i_{2k+4}$, $P^i_{2k+5}P^i_{2k+6}$, $\dots$, $P^i_{4k+1}P^i_{4k+2}$ (correspond to $q_9q_{10}$, $q_{11}q_{12}$, and $q_{13}q_{14}$ in $P_q$ in our example) are used to model membership of $x_i$ in the clauses $C_1$, $C_2$, $\dots$, $C_k$ (correspond to the clauses $A$, $B$, and $C$ in our example), respectively.

\item If $x_i$ appears in a clause $C_j$, then we do nothing for the edges correspond do those clauses.

\item Otherwise, if $x_i$ does not appear in $C_j$, then we introduce a \emph{flag vertex} in the graph, which is adjacent to $P^i_{2(k+j)+1}$ and $P^i_{2(k+j)+2}$.

\item Due to the rigidity of the backbone (up to $\varepsilon$ flexibility), this flag vertex lies on $\ell^C_j$.
Similarly, in our example, if $q$ appears in $B$, then $q_{11}q_{12}$ stays as is, but otherwise, a flag vertex is introduced, adjacent to both $q_{11}$ and $q_{12}$.

\item Every clause has 3 literals. Thus, on each horizontal line, 3 out of $m$ possible flag vertices will be missing.
That sums up to a total of $k(m-3)$ flag vertices for this part of the graph.

\item For the remaining sub-path $(P^i_1,\dots,P^i_{2k+2})$ of $P^i$ (corresponds to $(q_1,\dots,q_8)$ of $P_q$ in our example), we introduce the flag vertices for the pairs $(P^i_2,P^i_3)$, $(P^i_4,P^i_5)$, $\dots$, $P^i_{2k}$, $P^i_{2k+1}$ (correspond to $(q_6,q_7)$ $(q_2,q_3)$, $(q_4,q_5)$, $(q_6,q_7)$ in our example).

\item That is a total number of $k*m$ flag vertices for this part of the graph. In the whole graph, there are precisely $2km-3k$ flag vertices.
\end{itemize}

Realize that the embeddings on some vertical lines must be flipped upside-down to create space for the flag vertices.
This operation corresponds to the truth assignment of the literal that corresponds to that vertical line.
The configuration forces at least one literal to have a different truth assignment, because for a pair of symmetrical horizontal lines, say $\ell_A$ and $\ell'_A$, there must be at least one, and at most two missing flags for the disks to fit between $\ell_L$ and $\ell_R$.

The input graph, a YES-instance, and the realization of the YES-instance of the Monotone NAE3SAT formula $\Phi = (q \vee s \vee t) \wedge (q \vee r \vee t) \wedge (q \vee r \vee s)$ is given in Figure~\ref{fig:fullGadget}.

\begin{corollary} \label{cor:largestCycle}
	Given a graph $G = (V,E)$, deciding whether $G$ is a unit disk graph is an NP-hard problem when the size of the largest induced cycle in $G$ is of length $4$.
\end{corollary}

\begin{proof}
	In the NP-hardness proof given in Section~\ref{sec:nphard}, the largest cycles in the input graph are diamonds, and induced 4-cycles.
	The rest of the graph consists of long paths.
	Since axes-parallel unit disk graph recognition is a more restricted version of unit disk graph recognition problem, the claim holds. 
\end{proof}

\section{APUD(k,0) recognition is NP-complete} \label{sec:npcomplete}
In this section, we show that the recognition of axes-parallel unit disk graphs is NP-complete, when all the given lines are parallel to each other.
This version of the problem is referred to as $\APUD(k,0)$, as there are $k$ horizontal lines given as input, but no vertical lines.
We use the reduction given in Section~\ref{sec:nphard}.

\begin{theorem} \label{thm:apud0m}
$\APUD(k,0)$ recognition is NP-hard.
\end{theorem}

\begin{proof}
Consider the realization given in Figure~\ref{fig:skeletonRealization}.
Notice that the length of the paths $P^1, P^2, \dots, P^m$ ($P_q, P_r, P_s, P_t$ in our example), and thus the number of disks on vertical lines is equal.
Lemma~\ref{lem:3clauses4literals} (ii) implies that those disks must be centered between $\ell_\nabla$ and $\ell_\Delta$. Therefore, for the disks that correspond to the vertices on these paths, we do not need any vertical line. We can simply remove the vertical lines, and add an extra horizontal line for each clause.
For the disks that are adjacent to, but not on $P_L$ and $P_R$, we can simply add another horizontal line.
That is an extra horizontal line for each clause.
As a result, for a given instance $\Phi$ of Monotone NAE3SAT formula with $k$ clauses and $m$ literals; we have an instance $\Psi$ of $\APUD(2k+3,m+2)$ to prove NP-hardness with vertical lines, and an instance $\Psi'$ of $\APUD(4k+3,0)$ to prove NP-completeness without vertical lines.
For each clause, we have 3 horizontal lines.

In $\Psi'$, only disks that can ``jump'' from one horizontal line to another are the ones that are on the top line of $\ell^C_1$ and  bottom line of $\ell'^C_1$.
And those jumps do not change the overall configuration. 
\end{proof}

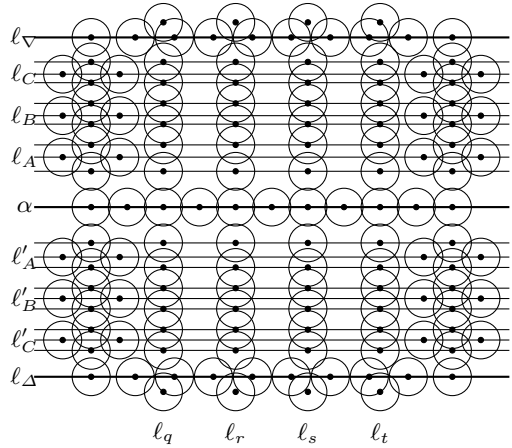
\begin{figure}[htbp]
\centering
\begin{tikzpicture}[scale=0.25]

\node[draw=none] at (-3.5,0) {$\alpha$};

\node[draw=none] at (3.8,-12) {$\ell_q$};
\node[draw=none] at (7.6,-12) {$\ell_r$};
\node[draw=none] at (11.4,-12) {$\ell_s$};
\node[draw=none] at (15.2,-12) {$\ell_t$};

\node[draw=none] at (-3.5,-9) {$\ell_\Delta$};
\node[draw=none] at (-3.5,-7.05) {$\ell'_C$};
\node[draw=none] at (-3.5,-4.85) {$\ell'_B$};
\node[draw=none] at (-3.5,-2.65) {$\ell'_A$};
\node[draw=none] at (-3.5,2.65) {$\ell_A$};
\node[draw=none] at (-3.5,4.85) {$\ell_B$};
\node[draw=none] at (-3.5,7.05) {$\ell_C$};
\node[draw=none] at (-3.5,9) {$\ell_\nabla$};

\tikzstyle{every node}=[draw, fill=black, shape=circle, minimum size=2pt,inner sep=0pt];

\draw[thick] (-3,0)--(22,0);

\foreach \i in {0,1.9,3.8,5.7,7.6,9.5,11.4,13.3,15.2,17.1,19} 
{
	\node at (\i,0) {};
	\draw (\i,0) circle[radius=1];			
}

\foreach \i in {3.8,7.6,11.4,15.2} 
{
	\foreach \j in {-7.6,-6.5,-5.4,-4.3,-3.2,-1.9,1.9,3.3, 4.4,5.5, 6.6,7.7}
	{
		\node at (\i,\j) {};
		\draw (\i,\j) circle[radius=1];		
	}
}

\foreach \i in {-7.6,-6.5,-5.4,-4.3,-3.2,-1.9,1.9,3.3, 4.4,5.5, 6.6,7.7} %
	\draw (-3,\i)--(22,\i);

	\foreach \i in {0,19} 
	{
		\foreach \j in {-7.6,-6.5,-5.4,-4.3,-3.2,-1.9,1.9,3.3, 4.4,5.5, 6.6,7.7}
		{
			\node at (\i,\j) {};
			\draw (\i,\j) circle[radius=1];
			
		}
	}


\foreach \i in {-7.05,-4.85,-2.65,2.65,4.85,7.05} 
{	
	\node at (1.5,\i) {};
	\node at (-1.5,\i) {};
	\node at (17.5,\i) {};
	\node at (20.5,\i) {};		
	\draw (1.5,\i) circle[radius=1];
	\draw (-1.5,\i) circle[radius=1];
	\draw (17.5,\i) circle[radius=1];
	\draw (20.5,\i) circle[radius=1];
	\draw (-3,\i)--(22,\i);
}

\draw[thick] (-3,9)--(22,9); 
\draw[thick] (-3,-9)--(22,-9); 

\foreach \i in {0, 2.32, 4.38, 6.43, 8.48, 10.53, 12.58, 14.63, 16.68, 19}
{
	\node at (\i,9) {};
	\draw (\i,9) circle[radius=1];
	\node at (\i,-9) {};
	\draw (\i,-9) circle[radius=1];
}

\foreach \i in {3.8,7.6,11.4,15.2} 
{
	\foreach \j in {-9.8,9.8}
	{
		\node at (\i,\j) {};
		\draw (\i,\j) circle[radius=1];
	}
}

\end{tikzpicture}
\caption{Realization of the graph given in Figure~\ref{fig:skeleton} on a set of parallel lines.}
\label{fig:allHorizontal}
\end{figure}

To show that $\APUD(k,0)$ recognition is in NP, we need to prove that a given solution can be verified in polynomial time as well as any feasible input will have a solution that takes up polynomial space, with respect to the input size.
Thus, we show that for any graph $G \in \APUD(0,k)$, there exists an embedding where the disk centers are represented using polynomially many decimals with respect to the input size.

For the following lemmas, define the set $\mathcal{H} = \{0, h_1, h_2, \dots, h_k\}$ where each element of the set corresponds to the Euclidean distance between a horizontal line, and the $x$-axis.
Without loss of generality, we assume that the bottom-most line is $x$-axis itself, and $h_i < h_j$ iff $i<j$.

\begin{defn}[Extension disk]
Let $A$ and $B$ a pair of intersecting unit disks centered at $(a,h_i)$ and $(b,h_j)$, respectively.
The \emph{extension disk} $A_\EX$ of $A$ is a disk centered at $(a,h_i)$, with radius $2$.
Then, $(b,h_j)$ is contained inside $A_\EX$, and symmetrically, $(a,h_i)$ is contained inside the extension disk $B_\EX$ of $B$.
\end{defn}

We give the trivial definition of extension disk for the sake of simplicity in the following lemmas.
Essentially, we claim that the center of a disk has some freedom of movement, and this freedom is determined by precisely two disks.
To show this dependence, we utilize the intervals on a line, defined by the intersection points between the extension disks with that line.

\begin{lemma}
Consider a disk $A$ centered at $(a,h_i)$. 
Each pair of extension disks, $B_\EX$ and $C_\EX$ that intersect $y = h_i$ line determines an interval $I_A$, in which the center of $A$ can move without changing the relationship among $A$, $B$, and $C$.
\end{lemma}

\begin{proof}
Let $\beta$ and $\gamma$ be two positive numbers, $B$ and $C$ be a pair of disks, and $B_\EX$ and $C_\EX$ be the extension disks, respectively.
Without loss of generality, let $(a-\beta, h_i)$ be the intersection point between $B_\EX$ and $y=h_i$ which is the closest intersection point to the center of $A$.
Similarly, let $(a+\gamma, h_i)$ be the intersection point between $C_\EX$ and $y=h_i$ which is the closest intersection point to the center of $A$.

If $A$ intersects both $B$ and $C$, then $I_A = [a-\beta, a+\gamma]$.
If $A$ intersects $B$ but not $C$, then $I_A = [a-\beta, a+\gamma)$.
If $A$ intersects $C$ but not $B$, then $I_A = (a-\beta, a+\gamma]$.
If $A$ intersects neither $B$ nor $C$, then $I_A = (a-\beta, a+\gamma)$. 
\end{proof}

\begin{lemma} \label{lem:interval}
Consider a disk $A$ centered at $(a,h_i)$, and assume that the Euclidean distance between any pair of parallel lines is different than $2$.
Let $I_A$ be an interval on $y = h_i$ line, in which the center of $A$ can move freely without changing the relationship between $A$ and any other disk.
Then, $I_A$ is determined by at most two extension disks that intersect $y = h_i$ line.
\end{lemma}

\begin{proof}
Let $p_1, p_2, \dots$ be the $x$-coordinates of the intersection points between $y=h_i$ line and extension disks that intersect $y = h_i$ line, such that $p_j < p_{j+1}$.
Suppose that $a \neq p_j$ for any $j$.
Then, there exists two intersection points whose $x$-coordinates are $p_j, p_{j+1}$ such that $p_j < a < p_{j+1}$ holds.
The interval $A_I$ is determined precisely by the extension disks whose intersection points are $(p_j,h_i)$ and $(p_{j+1}, h_i)$.

Observe that if $a = p_j$ holds for some $j$, this is the same case where $A$ intersects with the corresponding disk, say $B$.
Depending on the the center of $B$ being to the left or to the right of the center of $A$, the interval can still be defined by either the pair $p_{j-1}, p_j$ or $p_j,p_{j+1}$. 
\end{proof}

Now, let us show that the intervals are large enough if the disk centers have coordinates that are represented using polynomially many decimals.
We denote this interval by $I_A$ for the center of a disk $A$.
Let us refer to a number in the output as \emph{good number} if its representation has polynomially many decimals with respect to the input size, and refer to a number as \emph{bad number}, otherwise.

\begin{lemma} \label{lem:intersection}
Let $A$ be a disk centered at $(a,h_i)$, and let $B$ and $C$ be two disks that defines $I_A$.
If the centers of $B$ and $C$ are good numbers, then there exists a point $p$ in $I_A$, such that $p$ is also a good number.
\end{lemma}

\begin{proof}
Let $(b,h_j)$ and $(c,h_m)$ be the centers of $B$ and $C$, respectively.
Define $\delta_{ij} = |h_i - h_j|$, and define $\delta_{im} = |h_i - h_m|$.
The intersection points on the line $y = h_i$ are $(\sqrt{4-\delta_{ij}^2} - b, h_i)$ and $(\sqrt{4-\delta_{im}^2} - c, h_i)$.
And thus by Pythagorean theorem we have

\begin{align*}
&\left|(\sqrt{4-\delta_{ij}^2} - b) - (\sqrt{4-\delta_{im}^2} - c)\right|\\
=& \left|(b+c) + (\sqrt{4-\delta_{ij}^2} - \sqrt{4-\delta_{im}^2})\right|
\end{align*}
as the length of the interval $I_A$.

Since $b$ and $c$ are both good numbers, $b+c$ is also a good number.
So, even though the expression $(\sqrt{4-\delta_{ij}^2}- \sqrt{4-\delta_{im}^2})$ is a very small, it does not change the fact that there exists a number $p$ inside the interval $A_I$ that is represented with polynomially many decimals.
Therefore, $p$ is a good number. 
\end{proof}

\begin{lemma} \label{lem:representation}
The coordinates for the disk centers can be chosen in the form of $b + \Sigma_{i} \pm \sqrt{a_i}$, where $a_i = \sqrt{4-\delta_i^2}$ and $\delta_i$ is the distance between two of the given parallel lines, and $b$ is a good number.
\end{lemma}

\begin{proof}
Suppose that we shift all the disk centers as far left to right as possible.
This defines a partial order $D_1 \prec D_2 \prec \dots D_n$ of dependence among the disks.
Now, let us fix one disk, and then respect this partial order among other disks.

For each disk $D_{i}$ one of the following holds:
\begin{enumerate}[(i)]
\item The Euclidean distance between $D_i$ and $D_j$ is exactly $2$ for some $j < i$.
\item There are two other disks, $D_j$ and $D_m$ for some $j,m <i$ that defines the interval for the center of $D_i$ (by Lemma~\ref{lem:interval}).
\end{enumerate}

In case (i) holds, then the center of $D_i$ has the exactly same $x$-coordinate with the center of $D_j$.
In case (ii) holds, then $a_i = \sqrt{4-\delta_i^2}$, and by Lemma~\ref{lem:intersection}, the center of $D_i$ can have an $x$-coordinate which is a good number. 
\end{proof}

Now, we finally can show that $\APUD(k,0)$ is in NP by utilizing the previous lemmas.

\begin{lemma} \label{lem:decimals}
For every graph $G$ and a set of parallel lines $\mathcal{L}$, if $G$ can be realized on lines from $\mathcal{L}$ as unit disks, then there exists such a realization using polynomial number of decimals with respect to the input size. 
\end{lemma}

\begin{proof}
Since the input consists of only horizontal lines, $y$-coordinates of every disk center is fixed.
It remains to prove that $x$-coordinates of the disks in a solution cannot be forced to be in very small intervals.
That is, if a solution consists of disks whose centers are bad numbers, then we can perturb the solution to a new solution contains only good numbers.
The algorithm described below guarantees existence of an embedding of $G$ on parallel lines as unit disks, whose centers are good numbers.

\begin{enumerate}
\item Let $G_0, G_1, \dots, G_k$ denote the disjoint induced subgraphs of $G$, such that the vertices of $G_i$ correspond to the disks centered on line $y = h_i$.
\item Embed $G_0$ on $x$-axis with small perturbations which results as all disks on $y=0$ having polynomially many decimals.
\item For each $1 \leq i \leq k$; find an embedding of $G_i$ onto $y=h_i$ line by only considering neighbors from $\bigcup_{j<i}G_j$.
\end{enumerate}

Let us now show that the algorithm is correct.

Observe that each $G_i$ is a unit interval graph, and by \cite{unitIntervalGraphs}, we know that every unit interval graph has an embedding with polynomially many decimals.
Thus, we can embed $G_0$ on $x$-axis with small perturbations which results as all disks on $y=0$ having good number as their centers.

Now, consider $G_0 \cup G_1$.
If $h_1 > 2$, then $G_1$ is a standalone unit interval graph, without any relationship with $G_0$.
If, on the other hand, $h_1 \leq 2$, then by Lemma~\ref{lem:intersection} and Lemma~\ref{lem:interval}, we know that every vertex of $G_1$ can be embedded on $h_1$ as unit disks, of whose centers are good numbers.

For each $2 \leq i \leq k$, the algorithm processes every $G_i$ with respect to the previous embeddings, and thus, by Lemma~\ref{lem:interval}, $G_i$ has an embedding on $y = h_i$ as unit disks, whose centers are good numbers.
By Lemma~\ref{lem:representation}, we know that the perturbations can be done in a way such that every disk keeps its relationships with the other disks.
Note that the number of lines is bounded by the number of disks, as empty lines do not have any effect on the embedding.
Thus, even if we have the coordinates on the line $y = h_k$ as nested square roots, the resulting coordinates will have $a(1 + 2 + \dots + 2^{k})$ bits where $a$ is a good number.
Thus, the coordinates can be represented by polynomially many decimals, and hence are good numbers.

Therefore, if $G$ has an embedding on parallel lines, then we can obtain an embedding that can be represented by polynomially many decimals with respect to the input by starting from $G_0$, and gradually moving up to $G_k$, processing each subgraph iteratively. 
\end{proof}

\begin{corollary} \label{col:polytime}
Every yes-instance of $\APUD(k,0)$ has a polynomial witness that can be verified in polynomial time
\end{corollary}
\begin{proof}
By Lemma~\ref{lem:decimals}, we know that for every graph that has a feasible embedding on $k$ parallel lines, there exists an embedding that takes up polynomial space with respect to the input size.
Then, it remains to verify whether the solution satisfies the intersection graph, and the disks are centered on given lines, which is no more than $\mathcal{O}(n*k)$, where $n$ is the number of disks, and $k$ is the number of lines. 
\end{proof}

\begin{theorem} \label{thm:completeness}
$\APUD(k,0)$ is NP-complete.
\end{theorem}

\begin{proof}
Directly follows by Theorem~\ref{thm:apud0m}, Lemma~\ref{lem:decimals}, and Corollary~\ref{col:polytime}. 
\end{proof}

\section{APUD(1,1) recognition is open}
In this section, we discuss a natural basis for $\APUD(k,m)$ recognition problem.
That is, $k = 1$, and $m = 1$.
For the sake of simplicity, we can assume that given two lines are the $x$-axis and the $y$-axis.
First, we give some forbidden induced subgraphs for $\APUD(1,1)$.
Namely, those subgraphs are 5-cycle ($C_5$), a 4-sun ($S_4$), and a 5-star ($K_{1,5}$).

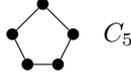
\begin{tikzpicture}[scale=0.2]
\node at (5,2) {$C_5$};

\tikzstyle{every node}=[draw, fill=black, shape=circle, minimum size=4pt,inner sep=0pt];
\node (A) at (0,4) {};
\node (B) at (2,2) {};
\node (C) at (1,0) {};
\node (D) at (-1,0) {};
\node (E) at (-2,2) {};
\draw (A)--(B)--(C)--(D)--(E)--(A);
\end{tikzpicture}
\hfill
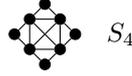
\begin{tikzpicture}[scale=0.2]
\node at (5,2) {$S_4$};

\tikzstyle{every node}=[draw, fill=black, shape=circle, minimum size=4pt,inner sep=0pt];
\node (A) at (0,4) {};
\node (B) at (1,3) {};
\node (C) at (2,2) {};
\node (D) at (1,1) {};
\node (E) at (0,0) {};
\node (F) at (-1,1) {};
\node (G) at (-2,2) {};
\node (H) at (-1,3) {};
\draw (A)--(B)--(C)--(D)--(E)--(F)--(G)--(H)--(A);
\draw (B)--(D)--(F)--(H)--(B);
\draw (B)--(F);
\draw (D)--(H);
\end{tikzpicture}
\hfill
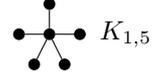
\begin{tikzpicture}[scale=0.2]
\node at (5,2) {$K_{1,5}$};
\tikzstyle{every node}=[draw, fill=black, shape=circle, minimum size=4pt,inner sep=0pt];
\node (A) at (0,2) {};
\node (B) at (0,4) {};
\node (C) at (2,2) {};
\node (D) at (1,0) {};
\node (E) at (-1,0) {};
\node (F) at (-2,2) {};
\draw (A)--(B);
\draw (A)--(C);
\draw (A)--(D);
\draw (A)--(E);
\draw (A)--(F);

\end{tikzpicture}

\begin{lemma} \label{lem:symmetry}
Consider two disks $A$ and $B$, centered on $(a,0)$ and $(b,0)$ with $a<0<b$.
If $|a| = |b|$, then another disk, $C$ that is centered on $(0,c)$ intersects either both, or none.
\end{lemma}

\begin{proof}
Consider the triangle whose corners are $(a,0)$, $(b,0)$, and $(0,c)$.
If $|a|=|b|$, then $\sqrt{a^2+c^2} = \sqrt{b^2+c^2}$ holds.
For $C$ to intersect $A$, $\sqrt{a^2+c^2} \leq 2$ must hold.
However, since $|a| = |b|$, if $\sqrt{a^2+c^2} \leq 2$ holds, then $\sqrt{b^2+c^2} \leq 2$ also holds.
Thus, $C$ intersects $A$ if, and only if $C$ intersects $B$.
For the same reason, if $\sqrt{a^2+c^2} > 2$ holds, then $\sqrt{b^2+c^2} > 2$ also holds.
Thus, $C$ does not intersect $A$ if, and only if $C$ does not intersect $B$. 
\end{proof}

\begin{lemma} \label{lem:c5s4k15}
	$C_5, S_4, K_{1,5} \not\in\APUD(1,1)$.
\end{lemma}

\begin{proof}
Let $A,B,C,D,E$ be the disks that form an induced 5-cycle on two perpendicular lines,
with the intersections between the pairs $(A,B), (B,C), (C,D), (D,E), (E,A)$.
By the pigeon hole principle, three of these disks must be on the same line.
Without loss of generality, let $A$, $C$, and $D$ be on $y=0$ line with centers $(a,0)$ and $(c,0)$, and $(d,0)$, respectively.
Up to symmetry, assume that $a < 0 < c < d$.
The remaining two disks, $B$ and $E$ must be centered on $x=0$.
Without loss of generality, assume that $e < 0 < b$.
However, $E$ cannot intersect $D$ without intersecting $C$ by Lemma~\ref{lem:triangle}.
Therefore, a 5-cycle cannot be realized as unit disks on two perpendicular lines.

Let $A,B,C,D$ be the disks that form the central clique of an induced 4-sun on two perpendicular lines with.
By Lemma \ref{lem:triangle}, two of these disks must be on one line, and the remaining two must be on the other line.
Without loss of generality, assume that $A$ and $C$ are on $y=0$ line, and $B$ and $D$ are on $x=0$ line.
Denote the centers of $A,B,C,D$ with $(a,0),(0,b),(c,0),(0,d)$, respectively, and assume that $a<0<c$ and $d<0<b$.
Let $X$ and $Y$ be two disks centered on two given perpendicular lines.
Assume that $X$ intersects $A,D$, and, $Y$ intersects $B,C$.
Clearly, $X$ and $Y$ should be on the same line, and on the different sides of the clique to avoid intersections with other rays.
By Lemma~\ref{lem:symmetry}, if $|b|=|d|$, then $X$ cannot intersect $B$ or $D$ independently.
By Lemma~\ref{lem:triangle}, if $|b|<|d|$, then $X$ cannot intersect $D$ without intersecting $B$. 
Similarly, if $|b| > |d|$, then $Y$ cannot intersect $B$ without intersecting $D$.
Therefore, a 4-sun cannot be realized as unit disks on two perpendicular lines.

Four rays $a,b,c,d$ of a 4-star $u;a,b,c,d$ must be on four different sides of the central vertex $u$.
To complete a 5-star, there must be one more ray, say $e$ centered on $(e,0)$.
If we can embed $e$ on one of the lines without intersecting any rays, then we can place another disk on $(-e,0)$ to form a $K_{1,6}$.
However, a $K_{1,6}$ cannot be realized as a unit disk graph.
Therefore, a 5-star cannot be realized as unit disks on two perpendicular lines. 
\end{proof}

\begin{lemma} \label{lem:bending}
A given graph $G$ can be embedded on $x$-axis and $y$-axis as a unit disk intersection graph, without using negative coordinates for the disk centers if, and only if $G$ is a unit interval graph.
\end{lemma}

\begin{proof}
In this proof, let us denote the class of graphs that can be embedded on $x$- and $y$- axes as unit disks, using positive coordinates only by ${(xy)}^+$.
We show that the disks on the $y$-axis can be rotated by $\pi/2$ degrees counterclockwise, and the intersection relationships can be preserved as given in $G$. 

$\mathbf{G \in {(xy)}^+ \Rightarrow G \in \mathbf{UIG}:}$
Consider two disks, $A$ and $B$, whose centers are $(a,0)$ and $(0,b)$, respectively, where $a$ and $b$ are both positive numbers.
If $A$ and $B$ do not intersect, then $\sqrt{a^2 + b^2} > 2$.
After the rotation, the center of $B$ will be on $(-b,0)$.
The new distance between the centers is $a + b$.
Since $(a+b)^2 > a^2 + b^2 > 4$, the inequality $a + b > 2$ holds.

If $A$ and $B$ intersect, then $\sqrt{a^2 + b^2} \leq 2$.
After the rotation, it might still be the case that $a + b > 2$.
However, now we can safely move the center of $B$ and the other centers that have negative coordinates closer to the center of $A$, recovering the intersection.
Note that if a disk $C$ is centered in between the centers of $A$ and $B$ after the rotation, then both $A$ and $B$ must intersect $C$ by Lemma~\ref{lem:triangle}.

$\mathbf{G \in \mathbf{UIG} \Rightarrow G \in {(xy)}^+:}$
Since $G$ is a unit interval graph, we can assume that every interval is a unit disk, and the graph is embedded on $x$-axis.
Consider two disks, $A$ and $B$, whose centers are $(-a,0)$ and $(b,0)$, respectively, where $a$ and $b$ are both positive numbers.
If $A$ and $B$ are intersecting, then $a + b \leq 2$
Then, after the rotation, since $a + b \leq 2$ holds, then $\sqrt{a^2 + b^2} \leq 2$ also holds.

If $A$ and $B$ are not intersecting, then $a + b > 2$.
After the rotation $\sqrt{a^2 + b^2} \leq 2$ might hold, creating an intersection between $A$ and $B$.
However, we can simply shift the center of $A$ (along with the other centers that are on $y$-axis) far away from the center of $B$, separating $A$ and $B$. 
\end{proof}

Lemma~\ref{lem:bending} shows that if we use only non-negative coordinates, then\\
$\APUD^{+}(1,1) = \APUD^{+}(1,0) = \UIG$.
This also applies if we use only non-positive coordinates.
Thus, a given $\APUD(1,1)$ can always be partitioned into two unit interval graphs.
Considering the embedding, one of these two partitions contains the disks that are centered on the positive sides of $x$- and $y$- axes, and the other partition contains the disks that are centered on the negative sides of $x$- and $y$- axes.

\begin{lemma} \label{lem:partition}
A graph $G \in \APUD(1,1)$ can be vertex-partitioned into four parts, such that any two form a unit interval graph.
\end{lemma}

\begin{proof}
Let $\Sigma(G)$ be an embedding of $G$ onto $x$- and $y$- axes as unit disks.
Denote the set of unit disks in $\Sigma(G)$ that are centered on the positive side of the $x$-axis, and the positive side of the $y$-axis by $\Sigma^+(G)$.
Similarly, denote the set of unit disks that are centered on the negative side of the $x$-axis, and the negative side of the $y$-axis $\Sigma^-(G)$.
By Lemma~\ref{lem:bending}, both $\Sigma^+(G)$ and $\Sigma^-(G)$ yield separate interval graphs.
The vertices that correspond to the disks in $\Sigma^+(G)$ yield a unit interval graph, as well as the vertices that correspond to the disks in $\Sigma^-(G)$.
Note that in case there exists a disk which is centered at $(0,0)$, then it can be included in one of the partitions arbitrarily.
Therefore, we can vertex-partition $G$ into two unit interval graphs. 
\end{proof}

Up to this point, we showed that if a unit disk graph can be embedded onto two orthogonal lines, then it can be partitioned into two interval graphs.
However, this implication obviously does not hold the other way around.
Thus, we now identify some structural properties of $\APUD(1,1)$.

\begin{rem} \label{rem:c4}
Consider four unit disks $A,B,C,D$, that are embedded onto $x$-axis and $y$-axis.
If they induce a $4$-cycle, then the centers of those disks will be at $(a,0)$, $(0,b)$, $(-c,0)$, $(0,-d)$, respectively, where $a,b,c,d$ are non-negative numbers.
\end{rem}

For the upcoming lemma, we will utilize the observation given in Remark~\ref{rem:c4}.
The lemma is an important step to describe a characterization of $\APUD(1,1)$.

\begin{lemma} \label{lem:twoc4s}
Consider eight unit disks embedded onto $x$-axis and $y$-axis, around the origin, whose intersection graph contains two induced $4$-cycles.
Then, this intersection contains at least four 4-cycles, each with a chord, not necessarily as induced subgraphs.
Moreover, those 4-cycles are formed by pairs of disks on the same direction ($+x$, $+y$, $-x$, $-y$) with respect to the origin.
\end{lemma}

\begin{proof}
Let those $4$-cycles be $(A,B,C,D)$, and $(U,V,W,X)$, in counterclockwise order, precisely as given in Remark~\ref{rem:c4}.
Consider the disks $A,B,C$ centered at $(a,0)$, $(0,b)$, $(-c,0)$, where $a,b,c$ are non-negative numbers.
Due to the configuration, $B$ intersects both $A$ and $C$, but $A$ and $C$ do not intersect.
If $a=0$, then by Lemma~\ref{lem:triangle}, $C$ intersect $A$ if it intersects $B$.
Thus, $a>0$, and by symmetry, $b,c>0$ holds.
Since $A$ intersects $B$, and $b>0$, then $a<2$ also holds.
Again, up to symmetry, $b,c < 2$ holds.

Now consider the disks $A,B,U,V$ centered at $(a,0)$, $(0,b)$, $(u,0)$, $(0,v)$, respectively.
With the same reasoning, $0<a,b,u,v<2$ holds.
Thus, the disks that are centered on the same axis intersect.
That is, $A$ and $U$, and, $B$ and $V$ intersect.
If $A$ and $V$ do not intersect, then $\sqrt{a^2 + v^2} > 2$.
Since $U$ and $V$ intersect, $\sqrt{u^2 + v^2} \leq 2$ holds, and $u < a$ should also hold.
Thus, by Lemma~\ref{lem:triangle}, $B$ and $U$ intersects.
Symmetrically, if $B$ and $U$ do not intersect, then $A$ and $V$ intersect.
With the same reasoning, since $B$ and $C$ intersect, $C$ and $V$ also intersect.
Therefore, these eight disks have four subsets, $\{A,B,U,V\}$, $\{B,C,V,W\}$, $\{C,D,W,X\}$, $\{D,A,X,U\}$, such that in each set, the induced graph is a $4$-cycle with at least one chord, which gives us four 4-cycles, each with a chord. 
\end{proof}

The given lemmas imply that for a connected graph $G \in \APUD(1,1)$, the following hold.
\begin{enumerate}[(i)]
\item $G$ does not contain either of 4-sun ($S_4$) or 5-star($K_{1,5}$) as an induced subgraph, and the largest induced cycle in $G$ is of length $4$  (by Lemma \ref{lem:c5s4k15}).
\item $G$ can be vertex partitioned into into four parts, such that any two form a unit interval graph (by Lemma~\ref{lem:partition}).
\item Given two $4$-cycles, $(a,b,c,d)$ and $(u,v,w,x)$ in $G$, each one of the quadruplets $\{a,b,u,v\}$, $\{b,c,v,w\}$, $\{c,d,w,x\}$, and $\{d,a,x,u\}$ is either a diamond or a $K_4$ (by Lemma~\ref{lem:twoc4s}).
\end{enumerate}

Although this characterization gives a rough idea about the structure of a graph $G \in \APUD(1,1)$, it is not clear whether the recognition can be done in polynomial time.
Note that the characterization is a necessary step through recognition, but it is not yet known whether it is sufficient.
Therefore, we conjecture that given a graph $G$, it can be determined whether $G \in \APUD(1,1)$ in polynomial time, and conclude our paper.

\section{Acknowledgments}
The author wants to thank Petr Hlin\v en\'y for his insight on the hardness proof. In addition, he thanks Deniz A\u{g}ao\u{g}lu  and Micha\l{} D\k{e}bski for their extensive comments and generous helps during the preparation of this manuscript.

	\bibliographystyle{plain}
\begin{tiny}
\bibliography{bibliography}	
\end{tiny}

\appendix

\end{document}